\documentclass[10pt,conference]{IEEEtran}

\usepackage{braket}
\usepackage{amsmath,amssymb,amsfonts,dsfont,amsthm}
\usepackage{algorithmic}
\usepackage{graphicx}
\usepackage{textcomp}
\usepackage{xcolor}
\usepackage{comment}
\usepackage[style=ieee]{biblatex}

\usepackage{lipsum}

\newcommand\blfootnote[1]{%
  \begingroup
  \renewcommand\thefootnote{}\footnote{#1}%
  \addtocounter{footnote}{-1}%
  \endgroup
}

\usepackage{flushend}

\newtheorem{theorem}{Theorem}
\newtheorem{corollary}[theorem]{Corollary}
\newtheorem{remark}[theorem]{Remark}

\newtheorem{definition}[theorem]{Definition}

\newlength{\blank}
\newcommand{\1}{\mathds{1}}
\DeclareMathOperator{\Tr}{Tr}

\newcommand{\EE}{\mathbb{E}}
\newcommand{\RR}{\mathbb{R}}

\newcommand{\cD}{\mathcal{D}}
\newcommand{\cE}{\mathcal{E}}
\newcommand{\cG}{\mathcal{G}}
\newcommand{\cL}{\mathcal{L}}

\newcommand{\cN}{\mathcal{N}}

\newcommand{\cS}{\mathcal{S}}

\newcommand{\cX}{\mathcal{X}}
\newcommand{\cY}{\mathcal{Y}}

\newcommand\myfrac[2]{\genfrac{}{}{0pt}{}{#1}{#2}}

\begin{document}

\title{Rate-Reliability Tradeoff for Deterministic Identification over Gaussian Channels}

\author{\IEEEauthorblockN{Pau Colomer\textsuperscript{1}}
\IEEEauthorblockA{\textit{Chair of Theoretical}\\
\textit{Information Technology}\\
\textit{Technische Universit\"at}\\
\textit{M\"unchen, Germany}\\
pau.colomer@tum.de}
\and
\IEEEauthorblockN{Christian Deppe\textsuperscript{2}}
\IEEEauthorblockA{\textit{Institute for Communications}\\
\textit{Technology}\\
\textit{Technische Universit\"at}\\
\textit{Braunschweig, Germany}\\
christian.deppe@tu-bs.de}
\and
\IEEEauthorblockN{Holger Boche\textsuperscript{2,3}}
\IEEEauthorblockA{\textit{Chair of Theoretical}\\
\textit{Information Technology}\\
\textit{Technische Universit\"at}\\
\textit{M\"unchen, Germany}\\
boche@tum.de}
\and
\IEEEauthorblockN{Andreas Winter\textsuperscript{1,4}}
\IEEEauthorblockA{\textit{Department Mathematik/Informatik}\\
\textit{Abteilung Informatik}\\
\textit{Universit\"at zu K\"oln}\\ 
\textit{Germany}\\
andreas.winter@uni-koeln.de}
}

\date{13 October 2025} 

\maketitle

\pagestyle{plain}

\begin{abstract}
We extend the recent analysis of the rate–reliability tradeoff in deterministic identification (DI) to general linear Gaussian channels, marking the first such analysis for channels with continuous output. 
Because DI provides a framework that can substantially enhance communication efficiency, and since the linear Gaussian model underlies a broad range of physical communication systems, our results offer both theoretical insights and practical relevance for the performance evaluation of DI in future networks. Moreover, the structural parallels observed between the Gaussian and discrete-output cases suggest that similar rate–reliability behaviour may extend to wider classes of continuous channels.
\end{abstract}


\blfootnote{The initial part (converse results) of this work has been accepted for presentation in the  \textit{2026 IEEE International Conference on Communications}, Glasgow, 24-28 May 2026 \cite{CDBW:Gaussian_ICC}. 

{\textsuperscript{1}Institute for Advanced Study, TUM, Garching, Germany.}\vspace{1pt}

{\textsuperscript{2}6G-life, 6G research hub,
Braunschweig and M\"unchen,
Germany.}\vspace{1pt}

{\textsuperscript{3}Cluster of Excellence, CeTI, Technische Universität Dresden, Germany.}\vspace{1pt}

{\textsuperscript{4}ICREA {\&} Universitat Aut\`onoma de Barcelona, Barcelona, Spain.}
}

\section{Introduction}
Modern communications are increasingly tasked to process bigger amounts of data, so improving the system efficiency becomes critical. In this context, the \emph{identification} paradigm, where a receiver is only concerned with checking if a particular message of his choosing is sent or not, arises as a promising model, especially for event-triggered and goal-oriented communications (molecular communication, digital watermarking, tactile internet, etc.) \cite{6G_Book,6G_Book_2}. The reason for this is that identification codes can achieve exponentially better performance than regular transmission of messages (the Shannon paradigm). Indeed, while the size of the $M$ transmittable messages grows linearly with the number $n$ of channel uses $\log M\sim nR$, it grows exponentially for the size of the $N$ messages that can be identified $\log N\sim 2^{nR}$ \cite{AD:ID_ViaChannels,HanVerdu:ID}. This result, however, is contingent on the use of randomisation on the encoder, which is not feasible for many applications. It is natural, then, to ask what happens for identification without randomisation also known as \emph{deterministic identification}.

\begin{definition}\label{def:DI_code}
A \emph{deterministic identification (DI)} code over $n$ uses of a channel $W:\cX\to\cY$ is a family of pairs $\{(u_i,\cD_i) : i=1,\dots,N\}$ with $u_i\in\cX^n$ input sequences and $\cD_i\subset\cY^n$ output regions such that for all $i\neq j\in [N]$:
\[
  W_{u_i}(\cD_i)\ge 1-\lambda_1 \quad\text{and}\quad W_{u_i}(\cD_j)\le\lambda_2,
\]
where $W_{x^n}(\cS):=W(\cS|x^n)$ is the probability of the channel outputting a sequence in $\cS\subset\cY^n$ when $x^n\in\cX^n$ was input. The conditions above bound the two errors that can be committed in an identification scheme: a \emph{missed identification}, when the receiver wants to identify the message sent but the outcome of the test is negative; and a \emph{wrong identification}, when the messages are different but the test is positive.
\end{definition}

Ahlswede, Dueck and Cai already noted \cite{AD:ID_ViaChannels,AC:DI}, and it was later proved in \cite{SPBD:DI_power}, that DI over discrete memoryless channels (DMCs) could only achieve linear-scaling codes $\log N_{DMC}\sim nR$ (albeit with higher rates than in transmission). However, it was also observed that for some continuous-alphabet channels the scaling could be superlinear, in fact \emph{linearithmic}: $\log N\sim R\,n\log n$ \cite{SPBD:DI_power,DI-poisson_mc,DI-fading}. This behaviour was later proven to be a universal characteristic of DI codes over general memoryless channels with continuous input and discrete outputs \cite{CDBW:DI_classical}, excitingly establishing a broad class of channels for which DI asymptotically outperforms any Shannon-style transmission scheme. This general work was subsequently complemented with a study of the behaviour of the DI rate in different error settings \cite{CDBW:Reliability_ICC,CDBW:Reliability-TCOM,DI-steins}.

DI over Gaussian channels falls outside the general framework of \cite{CDBW:DI_classical,CDBW:Reliability-TCOM} as they have continuous output. However, it has also been shown to exhibit a linearithmic scaling (actually it was the first channel for which this behaviour was observed in DI \cite{SPBD:DI_power,DI-fading}), with the corresponding linearithmic capacity bounded as $\frac38\leq\dot{C}_{DI}(\cG)\le\frac12$ (the lower bound is an improved new result from \cite{galaxy-codes}). All these works on Gaussian channels, however, contain only asymptotic results which tell us little about the behaviour of the errors beyond that they go to $0$. Still, note that by discretising the output of an infinite-alphabet channel, the achievability results from the discrete output case \cite{CDBW:DI_classical,CDBW:Reliability-TCOM} can be applied directly to yield typically linearithmic rate lower bounds.

In regular transmission of messages, it is known that when communicating at any rate below capacity, the error (incorrect decoding) vanishes exponentially fast with $n$. In contrast, it was shown in \cite{CDBW:Reliability-TCOM} that, for channels with finite output, if the identification errors vanish exponentially fast $\lambda_{1}\sim 2^{-nE_{1}}$, $\lambda_{2}\sim 2^{-nE_{2}}$, then we can no longer have linearithmic rate of DI, rather the rate is limited linearly, and proportional to $-\log\min\{E_1,E_2\}$. (In the converse, a similar result holds even when only one of the errors is exponentially small \cite{DI-steins}, and then the linear rate is upper-bounded proportional to $-\log\max\{E_1,E_2\}$.) In other words, we can only identify at linearithmic rate when the errors vanish more slowly than exponentially. In fact, we need $E_{1,2}(n) \gtrsim 1/n$. 
Does this curious rate-reliability tradeoff extend to 
infinite-output channels? Note that like in the regime of linearithmic rates, the achievability of linear rates with exponentially small errors can be obtained directly from discretising the output and applying the corresponding results of \cite{CDBW:Reliability-TCOM}. 

In the present paper, we initiate the rate-reliability study of DI over channels with infinite output, analysing the general linear Gaussian channel. We provide upper and lower bounds to the identification rates in terms of the errors which allow us observe and study different code size scaling regimes by fixing various error settings. Besides the strong theoretical motivation for such an analysis (evidenced in the previous paragraphs), there is also a practical one, rooted in the central role of Gaussian models as canonical abstractions for linear communication systems, sensor networks, and signal processing architectures, where rate-reliability functions are critical for performance assessment.


\medskip

\section{Channel model}
\label{sec:channel}
We will consider here the general \emph{linear Gaussian channel} $\cG:\RR^n\rightarrow \RR^n$, defined by the following input-output relation:
\begin{equation}\label{eq:input-output}
Y^n=A x^n+Z^n,
\end{equation}
where $x^n=x_1\dots x_n\in \RR^n$ is the input sequence (vector $x^n=\vec{x}$) fed into the channel, 
$A \in \RR^{n \times n}$ is a deterministic full-rank linear transformation ,
$Z^n \sim \cN(0, \Sigma)$ is an $n$-dimensional centred Gaussian-distributed noise vector with covariance matrix $\Sigma \in \RR^{n \times n}$, 
and $Y^n \in \RR^n$ is the channel output.
In other words, given an input sequence $x^n$, we have the following output probability distribution:
\begin{equation}
  \label{eq:outputPD}
  \cG_{x^n}:=\cG(Y^n|x^n)=\cN(Ax^n,\Sigma).
\end{equation}
As real communication systems cannot handle arbitrarily large signals, the following power constraint on the transmission power of each codeword is imposed:
\begin{equation}
  \label{eq:power_constraint}
  \|x^n\|^2\leq nP.
\end{equation}


To highlight the generality of the proposed linear Gaussian model, note that it subsumes many classical channels as special cases. The \emph{additive white Gaussian noise (AWGN) channel} appears for $A=\1_n$ and $\Sigma=\sigma^2\1_n$; the slow and fast fading models studied in \cite{DI-fading} correspond to $A=g\1_n$ and $A=\mathrm{diag}(g_1,\ldots,g_n)$, respectively. Allowing statistical dependence among the diagonal entries, we can model correlated fading; while choosing $A$ with Toeplitz structure captures inter-symbol interference (ISI). 
Non-square full-column-rank matrices $A \in \mathbb{R}^{m \times n}$ arise in MIMO and related settings involving projections, redundancy, or other dimension-changing effects. Such channels, however, offer no additional generality here, as any injective $A : \mathbb{R}^n \to \mathbb{R}^m$ is equivalent to an $n$-dimensional Gaussian channel obtained by projecting the output onto the $\mathrm{range}(A)$. For clarity we therefore restrict to the square, invertible case.


\medskip

\section{Preliminaries}
\label{sec:preliminaries}
Our method relies on studying metric properties of output probability sets conditioned on the input sequences. We include in this section the fundamental measures and relations that we will require.
Let us start simply by recalling the (squared) \emph{Euclidean distance} between two sequences $a^n,b^n\in\cX^n$:
\begin{equation}
    \|a^n-b^n\|^2 
      = {(a^n - b^n)^\top  (a^n - b^n)}
      = {\sum_{i=1}^n|a_i-b_i|^2}.
\end{equation}

The (squared) \emph{Mahalanobis distance}, corresponds to the Euclidean distance in a space transformed by a symmetric and positive semidefinite matrix $M$:
\begin{equation}
\label{eq:mahalanbobis}
    \|a^n-b^n\|_M^2 
      := (a^n - b^n)^\top M (a^n - b^n).
\end{equation}
As $M$ can be diagonalised (it is positive semidefinite), we can easily prove a relation with the Euclidean distance through its spectral decomposition $M=\sum_i\nu_i\vec{u_i}\vec{u_i}^\top$. Indeed,
\begin{equation}
  \label{eq:Mah-Euc}
  \nu_{\min} \|a^n-b^n\|^2
    \leq \|a^n-b^n\|_M^2
    \leq \nu_{\max} \|a^n-b^n\|^2.
\end{equation}
While the lower bound is trivial if $\nu_{\min}=0$, it is not for strictly positive definite matrices $M$ (i.e.~$\nu_{\min}>0$), which will be our case. 

We move now to statistical distance measures between two probability distributions with density functions $p$ and $q$ defined on $\RR^n$ (with respect to the Lebesgue meaaure). The \emph{total variation (TV) distance} is defined as
\begin{equation}
  \label{eq:TVD}
  \frac{1}{2} \|p-q\|_1 := \sup_{\cL \subseteq \RR^n} \big| p(\cL) - q(\cL) \big|
    = \frac{1}{2} \int_{\RR^n} |p(y) - q(y)| \, dy,
\end{equation}
where $p(\cL) = \int_{\cL} p(y) \, dy$ and $q(\cL) = \int_{\cL} q(y) \, dy$. Note that $\|\cdot\|_1$ is a norm.

The \emph{Bhattacharyya coefficient}, which we will call \emph{fidelity} in this paper, borrowing its name from quantum information, is defined as
\begin{equation}
  \label{eq:fidelitydef}
  F(p,q) := \int_{\RR^n} \sqrt{p(y) q(y)} \, dy.
\end{equation}
These last two satisfy the following relationship:
\begin{equation}\label{eq:FvdG}
    1 - F(p,q) \le \frac{1}{2} \|p-q\|_1 \le \sqrt{1 - F(p,q)^2},
\end{equation}
and it is known that $P(p,q) = \sqrt{1 - F(p,q)^2}$ is a metric on the space of probability distributions.


\section{Rate–Reliability Tradeoff Analysis}
\label{sec:results}
In this section, we aim to see the impact that a change on the reliability (the speed at which the errors go to zero) has to the rates of DI codes over general Gaussian channels. To this end, we write the errors in an exponential form $\lambda_{1,2}=e^{-nE_{1,2}(n)}$, where the subscript $\{\cdot\}_{1,2}$ means that the equation holds for both sub-indices separately. Our objective is to derive rate bounds that depend directly on the error exponents such that, by giving different values to the exponent (defining different reliabilities), we can observe how the rate changes. Notice also that the error exponents need to be at least of order $E_{1,2}(n)\geq\Omega(1/n)$ to keep the errors bounded $\lambda_{1,2}\ll 1$. 

We provide converse proofs in two different error settings. First, in Subsection \ref{ssec:converse}, we present a general case which is tightest for errors that are similar (symmetric). 
In \ref{ssec:asymmetric}, we study the maximally asymmetric error scenario (where errors are very dissimilar, in fact, only one of them is exponentially small). Finally, in Subsection \ref{ssec:achievability}, we discuss matching achievability rate results under different error conditions.

\subsection{Tradeoff in the symmetric error regime}
\label{ssec:converse}
We provide in this section an upper bound on the DI rate in terms of the minimum error exponent in Theorem \ref{thm:RR_upperbound}. The analysis of this bound yields critical rate results when the errors vanish exponentially fast or very slow.

\begin{theorem}
\label{thm:RR_upperbound}
For any DI code of block length $n$ over the general linear Gaussian channel described in section \ref{sec:channel}, and with positive error exponents, $E_1(n), E_2(n) \geq E(n) \geq \frac{\ln 16}{n}$, the linear rate $R(n)=\frac1n \log N$ is upper-bounded as
\[
  R(n) \leq \frac12\log\frac{8\nu_{\max}P}{E(n)},
\]
where $P$ is the transmission power constraint in Eq.~\eqref{eq:power_constraint}, $\nu_{\max}$ the maximum eigenvalue of the matrix $M=A^\top\Sigma^{-1}A$, and $A$ and $\Sigma$ are the linear transformation and the covariance that define the Gaussian channel (see Section \ref{sec:channel}).
\end{theorem}


\begin{proof}
By construction, any good DI code has pairwise total variation distance between two output distributions corresponding to different codewords lower bounded. Indeed, given an optimal decoding test $\cD_i\subset\RR^n$ for the codeword $u_i$ it is required by Definition \ref{def:DI_code} that $W_{u_i}(\cD_i)\geq 1-\lambda_1$ and $W_{u_{j}}(\cD_i)\le \lambda_2$ for $i\neq j$. Therefore, by the supremum definition of the TV distance [see Eq.~\eqref{eq:TVD}], it is clear that 
\begin{equation*}
  \label{eq:converseTVD}
  \frac12\|\cG_{u_i}-\cG_{u_{j}}\|_1
   \geq \frac12\|\cG_{u_i}(\cD_i)-\cG_{u_{j}}(\cD_i)\|_1
  \geq 1-\lambda_1-\lambda_2.
\end{equation*}

Let us now use the relations in Eq.~\eqref{eq:FvdG} to express the bound above in terms of the fidelity. We obtain:
\begin{equation}
  \label{eq:maxFidelity}
  F(\cG_{u_i},\cG_{u_{j}})^2
    \leq 1-\left(1-2e^{-nE(n)}\right)^2 
    < 4e^{-nE(n)}.
\end{equation}
Let us remark here that the above bound is tight only when the error exponents $E_1(n),\,E_2(n)$ are of similar order. If they are very different, a case studied in Subsection \ref{ssec:asymmetric}, it is loose. 

We have been motivated to move from total variation distance to fidelity in Eq.~\eqref{eq:maxFidelity} because the fidelity integral in Eq.~\eqref{eq:fidelitydef} for two displaced Gaussian distributions with the same covariance has a simple closed-form expression \cite{LieseVajda}. Indeed, let
\(
  p(\vec{y})=\mathcal N(\vec{y};A\vec{x},\Sigma)
\)
and
\(
  q(\vec{y})=\mathcal N(\vec{y};A\vec{x'},\Sigma),
\)
where recall that we are using a vector form for our sequences $x^n:=\vec{x}$ to make the notation more compact. Explicitly,
\begin{equation}
  \label{eq:MultiModalNormal}
  p(\vec{y})
    =\frac{\exp\!\left(-\frac12 (\vec{y}-A\vec{x})^\top\Sigma^{-1}(\vec{y}-A\vec{x})\right)}{(2\pi)^{n/2}(\det\Sigma)^{1/2}},
\end{equation}
with an analogous expression for $q(\vec{y})$. Therefore,
\begin{equation*}
  \label{eq:integrant}
  \sqrt{p(\vec{y})q(\vec{y})}
  \!=\! \frac{e^{-\frac14\big[(\vec{y}-A\vec{x})^\top\Sigma^{-1}(\vec{y}-A\vec{x})+(\vec{y}-A\vec{x'})^\top\Sigma^{-1}(\vec{y}-A\vec{x'})\big]}}{(2\pi)^{n/2}(\det\Sigma)^{1/2}}.
\end{equation*}
Defining $\vec{\mu}=(A\vec{x}+A\vec{x'})/2$ and $\vec{\Delta}=A\vec{x}-A\vec{x'}$, the exponent in the expression above becomes
\begin{equation}\label{eq:integrand_exponent}
-\frac12(\vec{y}-\vec{\mu})^\top\Sigma^{-1}(\vec{y}-\vec{\mu})-\frac18\,\vec{\Delta}^\top\Sigma^{-1}\vec{\Delta}.    
\end{equation}
Therefore, we can write the fidelity as
\begin{equation*}
\label{eq:fidelitycalc}
\begin{split}
  F(p,q) 
   :=& \int_{\RR^n}\sqrt{p(\vec{y})q(\vec{y})}\,d\vec{y}\\
    =& \exp\left(-\frac{\vec{\Delta}^\top\Sigma^{-1}\vec{\Delta}}{8}\right)\underbrace{\int_{\RR^n}\frac{e^{-\frac12 (\vec{y}-\vec{\mu})^\top\Sigma^{-1}(\vec{y}-\vec{\mu})}} {(2\pi)^{n/2}(\det\Sigma)^{1/2}}d\vec{y}}_{=1},
\end{split}
\end{equation*}
where the integral evaluates to $1$, because it is the volume of the normalised Gaussian $\cN(\vec{y};\vec{\mu},\Sigma)$. So, we have:
\begin{equation*}
\label{eq:fidelity_step}
\begin{split}
    F(p,q)&=\exp\!\left(-\frac{1}{8}(A\vec{x}-A\vec{x'})^\top\Sigma^{-1}(A\vec{x}-A\vec{x'})\right)\\
    &=\exp\!\left(-\frac{1}{8}(\vec{x}-\vec{x'})^\top(A^\top\Sigma^{-1}A)(\vec{x}-\vec{x'})\right)\\
    &=\exp\!\left(-\frac{1}{8}\|\vec{x}-\vec{x'}\|^2_{M}\right),
\end{split}
\end{equation*}
with $M:=A^\top\Sigma^{-1}A$ a symmetric and positive definite matrix (since $\Sigma$ is positive definite and $A$ has full rank). We use now the rightmost relation in Eq.~\eqref{eq:Mah-Euc} to bound the fidelity in terms of the Euclidean distance:
\begin{equation}
\begin{split}\label{eq:euc_F_bound}
-\frac{\nu_{\max}\|\vec{x}-\vec{x'}\|^2}{8} \leq -\frac{\|\vec{x}-\vec{x'}\|^2_{M}}{8}&=\ln F(p,q)\\
&<\ln2-\frac{nE(n)}{2},
\end{split}
\end{equation}
where in the last inequality we have included the upper bound on the fidelity that any two codewords have to respect, established in Eq.~\eqref{eq:maxFidelity}. Rearranging Eq.~\eqref{eq:euc_F_bound} above, we can obtain a minimum Euclidean distance at which any two codewords of a good DI code have to be. Indeed,
\begin{equation}
  \label{eq:packing_radius}
  \|\vec{x}-\vec{x'}\| > 2\sqrt{\frac{nE(n)-2\ln2}{\nu_{\max}}} =: 2r.
\end{equation} 
Here we invoke the lower bound on $E(n) \geq \frac{\ln 16}{n}$ to ensure that $r > 0$. 
We have found that the codewords of a good DI code have to form a packing of radius $r$. That is, around all potential codewords $u_i$ we can define an $n$-dimensional ball $\cS_{i}(n,r)$ of radius $r$ such that they do not pairwise intersect. In other words, the ball $\cS_{i}(n,r)$ corresponding to the codeword $u_i$ has to be disjoint from $\cS_{j}(n,r)$, the ball corresponding to any different codeword ($j\neq i$). Now, as all input codewords have to fulfil the power constraint $\|x^n\|^2\le nP$, the centres of the spheres (the codewords) have to be contained in the permissible set that the constraint defines. In our case, the power constraint defines an $n$-dimensional ball $\cS_0(n,\sqrt{nP})$ of radius $\sqrt{nP}$. Therefore, our problem is reduced to a geometric one: how many disjoint $\cS_{i}(n,r)$ can we pack with centres inside $\cS_0(n,\sqrt{nP})$? This will be an upper bound to the maximum number $N$ of different valid codewords.

\begin{figure}[ht]
    \centering
    \includegraphics[width=0.9\linewidth]{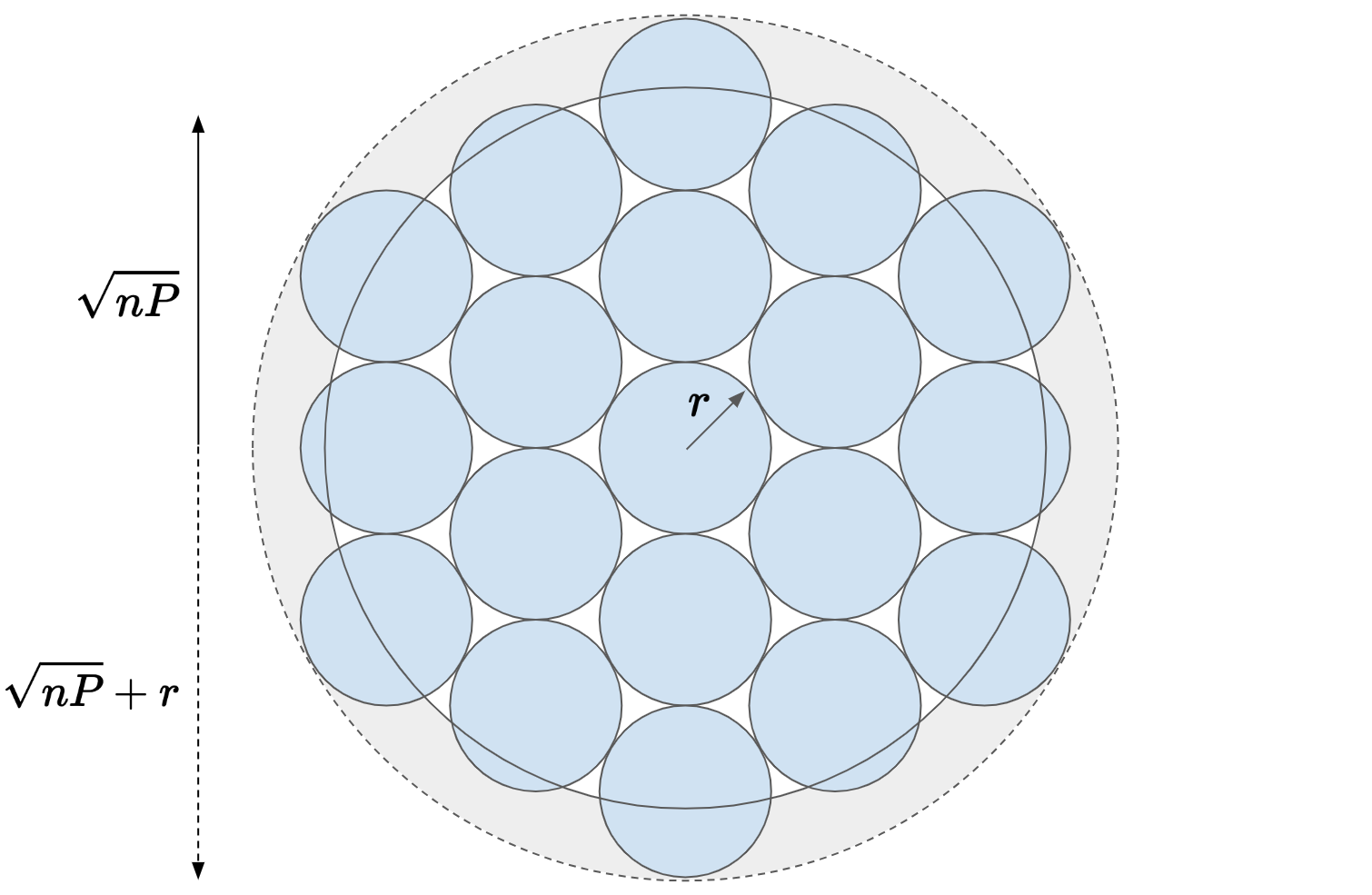}
    \caption{A packing of $N$ small spheres $\cS_i(n,r)$ centred inside the bigger $S_0(n,\sqrt{nP})$, and all contained by the biggest $S_+(n,\sqrt{nP}+r)$.}
    \label{fig:packing}
\end{figure}

The volume of a ball can be calculated analytically:
\begin{equation}
  \label{eq:volume_hypersphere}
  \textsc{Vol}[\cS(n,\rho)]
    = \frac{\pi^{n/2}}{\Gamma(\frac{n}{2}+1)}\rho^n,
\end{equation}
so we are motivated to use volumetric arguments to upper bound $N$. We can obtain such a bound by computing the ratio between the volume of a hypersphere $S_+(n,\sqrt{nP}+r)$ that contains all the balls $\cS_i(n,r)$ and the volume of one of these small $\cS_i(n,r)$, see Figure \ref{fig:packing}. We obtain:
\begin{equation*}
  \label{eq:N_upperbound}
  N \leq \frac{\textsc{Vol}[S_+(n,\sqrt{nP}+r)]}{\textsc{Vol}[S_i(n,r)]}
    = \frac{(\sqrt{nP}+r)^n}{r^n}\leq\left(\!\frac{2\sqrt{nP}}{r}\!\right)^n\!\!.
\end{equation*}
By introducing the definition of the linear rate $R(n)=\frac1n \log N$ and substituting the value of $r$ defined in Eq.~\eqref{eq:packing_radius} we get:
\begin{equation}
\label{eq:RR_upperbound}
\begin{split}
  R(n)
   &\leq\log\left(2\sqrt{nP}\right) - \log\sqrt{\frac{nE(n)-2\ln(2)}{\nu_{\max}}}\\
   &\leq\frac12\left[\log n -\log\left(nE(n)\right)+\log\left(8\nu_{\max}P\right)\right]\\
   &= \frac12\log\frac{8\nu_{\max}P}{E(n)},
\end{split}
\end{equation}
where in the inequality we have used that $nE(n)$ is at least twice as large as $2\ln 2$.
\end{proof}

Studying Theorem \ref{thm:RR_upperbound} 
we can easily conclude that the typical linearithmic behaviour of DI-codes over Gaussian channels is lost when the errors vanish exponentially fast, similarly to what we observed for channels with finite output \cite{CDBW:Reliability-TCOM}. Indeed, notice that imposing exponentially vanishing errors means fixing $E(n)=E>0$ constant. Then, directly from Eq.~\eqref{eq:RR_upperbound} we get:
\begin{equation}
  \label{eq:lost_linearithmicality}
  R(n) \leq \frac12\log\frac{2\nu_{\max}P}{E} = O(1),
\end{equation}
so, by inverting the definition of the rate, we get that the size of the $N$ identifiable messages $\log N=nR(n)\le O(n)$ can only increase linearly, completing the proof.

In contrast, when taking very slowly vanishing errors, with $E_{1,2}(n) \geq E(n) = \Omega(1/n)$, we can recover the best known upper bound on the DI capacity of the Gaussian Channel. Indeed, by introducing $E(n)=\Omega(1/n)$ in Eq.~\eqref{eq:RR_upperbound}, we find $R(n)\leq \frac12 \log n + O(1)$, recovering the linearithmic behaviour. Finally, inserting this rate into the definition of the linearithmic capacity:
\begin{equation}\label{eq:capacity}
\!\!\dot{C}_\text{DI}(\cG) 
    =\!\! \sup_{\myfrac{E_i(n)\rightarrow 0,}{nE_i(n)\rightarrow\infty}} \!\!
      \left( \liminf_{n\rightarrow\infty} \frac{1}{\log n} \max_{\myfrac{\text{errors }\lambda_i\text{ s.t.}}{\log \lambda_i \leq -nE_i(n)}}\!\! R(n) \right),
\end{equation}
we obtain
\begin{equation}
\dot{C}_\text{DI}(\cG)\leq\frac12+\liminf_{n\to\infty}\frac{\log(2\nu_{\max}P)}{2\log n}=\frac12,
\end{equation}
which is (if the reader recalls the Introduction) the best known capacity upper bound for the Gaussian channel.


In the analysis so far, only the smaller of the two error exponents, $E(n)=\min\{E_1(n),E_2(n)\}$, really plays a role. 
It is thus a priori possible that, to obtain linearithmic results, one of the two errors can vanish exponentially fast $E_i(n)=\Theta(1)$, as long as the other is $E_{j\ne i}(n)=\Omega(1/n)$ (with $i,j=\{1,2\}$). As we have already noted in the proof of Theorem \ref{thm:RR_upperbound}-- recall the discussion after Eq.~\eqref{eq:maxFidelity}-- the above method is tighter when the error exponents are similar, and looser when they are very distinct, and not applicable when the smaller exponent is $0$.
In the following subsection, we develop a different method to study maximally asymmetric error settings, which will give us the conditions that both error exponents (and not just the minimum one) have to fulfil to find linearithmic codes.

\medskip
\subsection{Tradeoff in asymmetric error regimes}
\label{ssec:asymmetric}
We study the following \emph{maximally asymmetric} settings:
\begin{itemize}
    \item \textbf{Stein regime:} $\lambda_1<1$ and $\lambda_2=2^{-nE_2}$,
    \item \textbf{Sanov regime:} $\lambda_1=2^{-nE_1}$ and $\lambda_2<1$.
\end{itemize} 
These two regimes mirror two well-known classical asymmetric hypothesis testing settings. So we will borrow tools from that context, drawing inspiration from \cite{DI-steins}. 

We present an upper bound for the regimes above in Theorem \ref{thm:stein_upperbound} and analyse its implications in the subsequent discussion.

\begin{theorem}
\label{thm:stein_upperbound}
For any DI code of block length $n$ over the general linear Gaussian channel described in section \ref{sec:channel}, in the Sanov (take $E_1$ below) or the Stein (take $E_2$) regimes, the linear rate is upper-bounded as follows:
\[
  R(n) \leq \log\left(\sqrt{\frac{2\nu_{\max}P}{E_{1,2}(n)}}+1\right),
\]
where $P$ is the transmission power constraint, $\nu_{\max}$ the maximum eigenvalue of the matrix $M=A^\top\Sigma^{-1}A$, with $A$ the linear transformation and $\Sigma$ the covariance that define the Gaussian channel.
\end{theorem}

Note that compared to Theorem \ref{thm:RR_upperbound}, the upper bound is slightly worse because of the $+1$ term, but it is much better in the sense that there, the minimum of the exponents $E_1(n),\,E_2(n)$ figures in the denominator, whereas here we may choose the maximum.

\begin{proof}
Let us start with the quantity that is usually employed to characterise the tradeoff between the errors in finite block length asymmetric hypothesis testing settings: the \emph{hypothesis testing relative entropy}
\begin{equation}\label{eq:HTREntopy}
  D_h^\epsilon(P\|Q) 
    := -\log\inf_{\cL\subset\cY \text{ s.t.}\atop P(\cL^c)\leq\epsilon} Q(\cL),
\end{equation}
where $\cL$ is a decision region in the output $\cY=\RR^n$, and $\cL^c = \RR^n\setminus\cL$ its set complement. 

Now, take any good DI code in the Stein regime. The identification test $\cE_j$ satisfies $\cG_{u_j}(\cE_j^c) \leq \lambda_1$ and $\cG_{u_k}(\cE_j) \leq \lambda_2$ (for $k\neq j$). Hence, we get immediately, by definition,
\begin{equation*}
  D_h^{\lambda_1}(\cG_{u_j}\|\cG_{u_k}) \geq 
  -\log\lambda_2 = nE_2.
\end{equation*}
Similarly, by taking $\cE_k$ as the test, we get $D_h^{\lambda_2}(W_{u_j}\|W_{u_k}) \geq nE_1$ for a DI code in the Sanov regime. In other words, all good DI codes in the Stein or Sanov regimes have output distributions with a minimum pairwise hypothesis testing relative entropy linear in $n$. This property has a similar flavour to the main argument we used in the symmetric error setting, and indeed, we aim now to somehow translate this object into a bound on the minimum Euclidean distance between two codewords.

We start by transforming the entropies above into R\'enyi relative entropies, which are defined, given a parameter $\alpha\in(0,\infty)\setminus\{1\}$ to be fixed later, as:
\begin{equation}\label{eq:Renyi_def}
D_\alpha(P\|Q)=\frac{1}{\alpha-1}\log\int_{\cY}P(y)^\alpha Q(y)^{1-\alpha}dy.
\end{equation}
This is connected to the hypothesis testing relative entropy, for $\alpha>1$, through the following bound \cite{CMW16:Dh-to-D_alpha}:
\begin{equation*}
\!\!D_h^{\epsilon}(W_{u_j}\|W_{u_k})\!\leq\! D_\alpha (W_{u_j}\|W_{u_k})+\frac{\alpha}{1-\alpha}\log\left(\!\frac{1}{1-\epsilon}\!\right).
\end{equation*}
Therefore, for both the Stein and Sanov regimes, we can bound
\begin{equation}
  \label{eq:Renyi_converse_bound}
  D_\alpha(\cG_{x^n}\|\cG_{x'^n})
    \geq nE_{1,2}-\frac{\alpha}{1-\alpha} \log\left(\frac{1}{1-\lambda_{2,1}}\right).
\end{equation}

The R\'enyi divergence above has a closed-form expression, generalising the fidelity \cite{LieseVajda}: indeed, with $\cG_{\vec{x}}(\vec{y})=\cN(y;A\vec{x},\Sigma)$ and $\cG_{\vec{x'}}(\vec{y})=\cN(y;A\vec{x'},\Sigma)$ [see Eq.~\eqref{eq:MultiModalNormal} for the density] we have that the integrand of the R\'enyi entropy in Eq.~\eqref{eq:Renyi_def}, $\cG_{\vec{x}}(\vec{y})^\alpha \cG_{\vec{x'}}(\vec{y})^{1-\alpha}$, equals:
\begin{equation*}
  \frac{e^{-\frac{\alpha}{2} (\vec{y}-A\vec{x})^\top\Sigma^{-1}(\vec{y}-A\vec{x})+\frac{\alpha-1}{2}(\vec{y}-A\vec{x'})^\top\Sigma^{-1}(\vec{y}-A\vec{x'})}}{(2\pi)^{n/2}(\det\Sigma)^{1/2}}.
\end{equation*}
Similarly to what we did in Eq.~\eqref{eq:integrand_exponent} we can define $\vec{\mu}=\alpha A\vec{x}+(1-\alpha)A\vec{x'}$ and $\vec{\Delta}=A\vec{x}-A\vec{x'}$ and the exponent above becomes:
\begin{equation}
  \label{eq:integrand_exponent2}
  -\frac12(\vec{y}-\vec{\mu})^\top\Sigma^{-1}(\vec{y}-\vec{\mu})-\frac{\alpha(1-\alpha)}{2}\|\vec{x}-\vec{x'}\|_M^2.
\end{equation}
Now, the integral corresponding to the first term of the expression above exactly cancels the prefactor $(2\pi)^{-n/2}(\det\Sigma)^{-1/2}$ as it corresponds to a Gaussian centred at $\vec{\mu}$ with covariance $\Sigma$, therefore the integral over $\RR^n$ of Eq.~\eqref{eq:integrand_exponent2} is exactly $\exp\left(-\frac{\alpha(1-\alpha)}{2}\|\vec{x}-\vec{x'}\|_M^2\right)$. Substituting into Eq.~\eqref{eq:Renyi_def} we get:
\begin{equation*}
  D_\alpha(P\|Q)=\frac{1}{\alpha-1}\frac{\alpha(\alpha-1)}{2}\|\vec{x}-\vec{x'}\|_M^2=\frac{\alpha}{2}\|\vec{x}-\vec{x'}\|_M^2.
\end{equation*}
Putting together the result above, Eq.~\eqref{eq:Renyi_converse_bound}, and the upper bound on Eq.~\eqref{eq:Mah-Euc} we find:
\begin{equation}
  \label{eq:stein_radius}
  \|\vec{x}-\vec{x'}\|^2\geq\frac{2}{\nu_{\max}}\left[\frac{nE_{1,2}}{\alpha}-\frac{1}{1-\alpha}\log\left(\frac{1}{1-\lambda_{2,1}}\right)\right].
\end{equation}
We can now optimise this minimum distance over the parameter $\alpha\in(1,\infty)$. The function $f(\alpha)=(\frac{C}{\alpha}+\frac{K}{\alpha-1})^{1/2}$ with $C,K>0$ two positive constants is convex on alpha with the minimum on $f_{\min}=\sqrt{C}+\sqrt{K}$. Therefore, by choosing $C:=\frac{2nE_{1,2}}{\nu_{\max}}$ and $K:=\frac{2}{\nu_{\max}}\log\left(\frac{1}{1-\lambda_{2,1}}\right)$, we finally find a lower bound on the Euclidean distance between any two codewords of a good DI code in the Stein or Sanov regimes:
\begin{equation}
\begin{split}
  \|\vec{x}-\vec{x'}\|
    &\geq \sqrt{\frac{2}{\nu_{\max}}}\left[nE_{1,2}+\log\left(\frac{1}{1-\lambda_{2,1}}\right)\right]^{\frac12}\\
    &>    \sqrt{\frac{2nE_{1,2}}{\nu_{\max}}} =: 2r.
\end{split}
\end{equation}

Similarly to the symmetric error case, we have obtained that the codewords of any good DI code in the Stein (take $E_2$ and $\lambda_1$ in the equation above) or Sanov (take $E_1$ and $\lambda_2$) regimes have to form a packing of radius $r$. The argument becomes now the same as in the symmetric error case: we will reuse the volumetric tools (recall Fig.~\ref{fig:packing}) to upper-bound the size $N$ of the packing, obtaining
\begin{equation}
  N\leq\left(\frac{\sqrt{nP}+r}{r}\right)^n
    =\left(\sqrt{\frac{2\nu_{\max}P}{E(n)}}+1\right)^n.
\end{equation}
Therefore, the linear rate $R(n)=\frac1n \log N$ satisfies
\begin{equation}
  \label{eq:stein_RR}
  R(n) \leq \log\left(\sqrt{\frac{2\nu_{\max}P}{E_{1,2}(n)}}+1\right),
\end{equation}
and the proof is complete.
\end{proof}

Studying Theorem \ref{thm:stein_upperbound} similarly to how we studied Theorem \ref{thm:RR_upperbound} in the previous section we obtain the following conclusions:
first, the typical linearithmic rate behaviour of DI-codes over the Gaussian channel is also lost when only one of the errors is exponentially small, and replaced by linear rate scaling. Indeed, for either one $E_{1,2}(n)=E>0$ a positive constant, the linear rate in Eq.~\eqref{eq:stein_RR} is upper-bounded by $\frac12\log\frac{1}{E} + O(1)$. Second, and in contrast, when both error exponents $E_{1,2}(n) = \Omega(1/n)$, we recover the linearithmic behaviour of the rate.

\medskip
\subsection{Achievability of linear rate and exponentially small errors}
\label{ssec:achievability}
So far in this paper, we have only discussed converse results of the rate-reliability functions for DI over general linear Gaussian channels. We discuss in this section, for completeness, the achievability part. Indeed, a lower bound is needed to demonstrate that the upper bounds above and the conclusions we extracted are actually tight at leading order. 

We should demonstrate two things. First, that we can construct linear scale codes with exponentially small errors. Indeed, we found an $O(1)$ upper bound on the linear rate when errors vanish exponentially fast, but for the existence of such a code we need an $\Omega(1)$ lower bound. And secondly, that we can construct linearithmic scale codes whose errors vanish slowly enough.

The first task above can be shown immediately from previous results. We just need to take a linear rate transmission code for a Gaussian channel, which are well-studied \cite{CT_book,Gallager:Gauss_chapter}. Then notice first that when transmitting below capacity, the error of a transmission code may vanish exponentially fast, $\lambda=2^{-nE_T}$, and secondly that any transmission code is automatically a DI code. Indeed, transmission codes are deterministic (as a matter of fact, it is well known that randomness does not help in transmission of messages), and if we can decode a message, we can of course identify it with probability of missed and false identification bounded by $\lambda_1, \lambda_2 \leq \lambda = 2^{-nE_T}$. 

The second task is also partially solved, as the achievability in the linearithmic scenario was proven in \cite{SPBD:DI_power,DI-fading} for restricted Gaussian channels (AWGN and fading channels). Indeed, linearithmic codes are constructed in these references with polynomially vanishing error exponents $E_1(n), E_2(n) = \omega(1/n)$. However, the general case of DI over general linear Gaussian channels has not been proven. 

We provide in this section a code construction which will show not only the existence of linear scale codes for DI over general linear Gaussian channels, but also the trade-off between the error exponents and the linear rate, which turns out to go to $\infty$ with the vanishing of the former, cf.~\cite{CDBW:Reliability-TCOM}, indicating its faster-than-linear scaling nature (linearithmic, in fact). Indeed, we will study the conditions the errors have to fulfil to achieve linearithmic code sizes, and how this scaling is lost for exponentially vanishing errors, matching the converse results developed in the previous subsections. 

The code will be constructed in a very classical manner: we pick the code words as elements of a packing (with a minimum pairwise distance) and then we use distance decoding on the output. That is, if the distance between the output sequence and the expected mean of the word we are testing is close enough, we will accept it; if it is too large, we will reject it. This is exactly the coding idea used for the initial construction of a DI code for AWGN channels \cite{SPBD:DI_power} and fading channels \cite{DI-fading} and, while there are different methods (like galaxy coding \cite{galaxy-codes}) that can be tighter, distance decoding will suffice in this article. Indeed, galaxy codes can improve the prefactor multiplying the first order of the rate in a linearithmic code, but here we are just interested in the scaling of the rate, independent of the prefactor.

Let us start with a lower bound on the size of a maximal packing of spheres $\cS_i(n,r)$ of radius $r$ with centres in the sphere $\cS_-(n,\sqrt{nP}-r)$, which is the input space allowed by the power constraint. We use volumetric arguments:
\begin{equation}
\label{eq:code_size}
\begin{split}
N&=\frac{\textsc{Vol}\left[\bigcup_{i=1}^N\cS_i(n,r)\right]}{\textsc{Vol}\left[\cS_i(n,r)\right]}\\
&\geq\frac{\Delta_n \textsc{Vol}\left[\cS_-(n,\sqrt{nP}-r)\right]}{\textsc{Vol}\left[\cS_i(n,r)\right]}\\
&\geq2^{-n}\left[\frac{\sqrt{nP}-r}{r}\right]^n=\left(\frac{1-\epsilon}{2\epsilon}\right)^n,
\end{split}
\end{equation}
where $\Delta_n$ is the density of the packing which is lower bounded by $\Delta_n\geq2^{-n}$ by standard volumetric arguments \cite{CS:Packings_lattices} (we refer the reader to other DI-code constructions \cite{SPBD:DI_power,DI-fading,CDBW:DI_classical}, which also use this step, for a more detailed development); and where in the last equality we have defined $r:=\epsilon\sqrt{nP}$. Now, notice that for $\epsilon\in(0,\frac13)$ a constant, meaning that the radius of our packing $r=\Theta(\sqrt{n})$, the code can just be linear with rate $R:=\frac1n\log N\leq \log\frac{1-\epsilon}{2\epsilon}$, a finite and positive constant. In contrast, if $\epsilon(n)\geq 0$ is a decreasing function of $n$ we can find higher order scaling.


Next, let us define for each message $i\in[N]$ our distance decoder as follows:
\begin{equation}\label{eq:decoding_region}
 \cD_i = \left\{ \vec{y}\in\cY^n:\|\vec{y}-A\vec{u_i}\|^2\leq\Tr\Sigma+4\nu_M n\sqrt{E_1(n)} \right\},
\end{equation}
with $\nu_M$ the maximum eigenvalue of the covariance matrix $\Sigma$, and the first type of error written in the exponential form $\lambda_1^{-nE_1(n)}$ and $E_1(n)\leq1$ (in Corollary \ref{cor:big_exp} we will pick a slightly different decoder to study the big exponent regime $E_1>1$). The receiver will identify the codeword $\vec{u_i}$ if the output sequence fulfils the distance relation above with respect to the expected output when $\vec{u_i}$ is input. Now, a missed identification will happen for all $i\in[N]$ when
\begin{align}
P_{e,1}(i) 
 &=\Pr\left\{ \|Y^n-A\vec{u_i}\|^2 > \Tr\Sigma+4\nu_M n\sqrt{E_1(n)}\middle|\vec{u_i}\right\}\nonumber \\
 &\leq\Pr\left\{\|Z^n\|^2>\Tr\Sigma+\delta_n \right\},\label{eq:E1}
\end{align}
with $Z^n\sim\cN(0,\Sigma)$ as described by the input-output relation in Eq.~\eqref{eq:input-output}, and $\delta_n\leq4\nu_M nE_1(n)$ to be fixed later. We aim to bound the probability above using a Chernoff bound. Before doing so, notice that we can write the random variable $Z^n=\Sigma^{1/2}L^n$ with $L^n\sim\cN(0,\1_n)$, such that $S:=\|Z^n\|^2=(L^n)^T\Sigma L^n=\sum_{i=1}^n\nu_i L_i^2$, with $\nu_1,\dots,\nu_n>0$ the eigenvalues of the covariance matrix. To alleviate the notation in the following calculations let us define $\mu:=\EE(S)=\Tr \Sigma =\sum_{i=1}^n\nu_i$, $v=\Tr \Sigma^2 =\sum_{i=1}^n\nu_i^2$, and $\nu_M=\max_i\nu_i$. We start by calculating the moment generating function of the centred random variable $S$. For any $0<s<\frac{1}{2\nu_M}$:
\[
\EE\!\left[e^{s(S-\mu)}\right]\!
=\!e^{-s\mu}\prod_{i=1}^n\EE\left[e^{s\nu_i(L_i^2)}\right]\!
=e^{-s\mu}\prod_{i=1}^n\left(1-2s\nu_i\right)^{-\frac12}\!,
\]
where the last equality follows from the standard Gaussian integral resulting from the expectation value. Taking now the logarithm and using the property $-\ln(1-x)\leq x+\frac{x^2}{2(1-x)}$ (valid for $0\leq x<1$, and therefore for $x=2s\nu_i$ for all valid values of $s$) we obtain:
\[
\begin{split}
\ln \EE\left[e^{s(S-\mu)}\right]&=-s\mu-\frac12\sum_{i=1}^n \ln(1-2s\nu_i)\\
&\leq-s\mu+s\mu+s^2\sum_{i=1}^n\frac{\nu_i^2}{1-2s\nu_i}\\
&\leq\frac{s^2v}{1-2s\nu_M}.
\end{split}
\]
Now we are ready to apply the Chernoff bound:
\begin{equation*}
\begin{split}
\Pr \{\|Z^n\|^2>\Tr&\Sigma+\delta_n \}
=\Pr\left\{S-\mu>\delta_n \right\}\\
&=\inf_{0<s<\frac{1}{2\nu_M}}\Pr\left\{e^{s(S-\mu)}>e^{s\delta_n} \right\}\\
&\leq\inf_{0<s<\frac{1}{2\nu_M}} e^{-s\delta_n}\EE\left[e^{s(S-\mu)}\right]\\
&\leq\inf_{0<s<\frac{1}{2\nu_M}} \exp\left[-s\delta_n +\frac{s^2v}{1-2s\nu_M}\right].
\end{split}
\end{equation*}
By standard critical point calculation, we observe that when choosing $\delta_n=\frac{2sv(1-s\nu_M)}{(1-2s\nu_M)^2}$ the exponent above is minimum over all values of $s\in(0,\frac{1}{2\nu_M})$ and can take any value in $(0,\infty)$. Then, we can choose the negative exponent to be $nE_1(n)>0$:
\begin{equation*}
   \Pr\left\{S-\mu>\delta_n \right\}\leq\exp\left[-\frac{s^2v}{(1-2s\nu_M)^2}\right]:=e^{-nE_1(n)}.
\end{equation*}
Putting together the equation above and the value we gave to $\delta_n$ to minimize the exponent, one observes the following relation:
\begin{equation}\label{eq:delta_value}
    \delta_n=2\nu_Mn\left[E_1(n)+\sqrt{E_1(n)}\right]\leq4n\nu_M\sqrt{E_1(n)},
\end{equation}
where the last inequality is true if $E_1(n)\leq1$, we are interested in this regime because we aim to study what happens for slowly vanishing errors (small error exponents). See Remark \ref{cor:big_exp} for further information on the regime $E_1\geq1$. 
Finally, by invoking Eq.~\eqref{eq:E1}, we get:
\begin{equation}\label{eq:Pe1_bound}
\begin{split}
P_{e,1}(i):&=\Pr \left\{\|Z^n\|^2>\Tr\Sigma+4n\nu_M\sqrt{E_1(n)} \right\}\\
&<\Pr \{\|Z^n\|^2>\Tr\Sigma+\delta_n \}\leq e^{-nE_1(n)}.
\end{split}
\end{equation}
The equation above directly relates the first type of error with the size of the decoding region. By choosing different values of the error exponent, we can now see the tradeoff between these two objects in our construction.

Let us work similarly for the other error: a false identification. For any two different messages $i\neq j\in[N]$ with a linearly transformed distance vector between the corresponding codewords defined as $\vec{d}:=A(\vec{u_j}-\vec{u_i})$ a false identification will happen when
\begin{align}
P_{e,2}(i,j)\!&=\! \Pr\left\{\|Y^n-A\vec{u_i}\|^2\!\leq\!\Tr\Sigma+4\nu_Mn\sqrt{E_1(n)}\middle|\vec{u_j}\right\} \nonumber\\
&=\!\Pr\left\{\|Z^n-\vec{d}\|^2\!\leq\!\Tr\Sigma+4\nu_Mn\sqrt{E_1(n)}\right\}.\label{eq:Pe2}
\end{align}
We start, as we did before, by defining the random variable $R:=\|Z^n-\vec{d}\|^2$ with mean $\mu:=\EE[R]=\Tr\Sigma+\|\vec{d}\|^2$, the diagonalization of the covariance matrix $\Sigma=U\Lambda U^T$ with $\Lambda=\text{diag}(\nu_1,\dots,\nu_n)$, and the distance vector in the diagonal basis $\vec{b}$ with coordinates $b_i:=\frac{1}{\sqrt{\nu_i}}(U^T\vec{d})_i$. We keep the definitions of $v:=\Tr \Sigma^2$ and $\nu_M=\max_i\nu_i$, and we redefine the random variable $L^n:=U^T\Sigma^{-1/2}Z^n$ so $L^n\sim\cN(0,\1_n)$. 
Notice that the change of basis using $U$ for the variable $Y$ was not needed for the first type of error but it is critical now so that we can write $R=\sum_{i=1}^n\nu_i(L_i-b_i)^2$. Now, the logarithm of the moment generating function of the centred random variable $R$ is for all $s>0$:
\[
\begin{split}
\ln\EE[e^{-s(R-\mu)}]&=s\mu+\ln\prod_{i=1}^n\EE[e^{-s\nu_i(L_i-b_i)^2}]\\
&=\sum_{i=1}^n\left[s\nu_i+\frac{2s^2\nu_i^2b_i^2}{1+2s\nu_i}-\frac12\ln(1+2s\nu_i)\right]\\
&\leq\sum_{i=1}^n\left[\frac{2s^2\nu_i^2}{1+2s\nu_i}(1+b_i^2)\right]\\
&\leq2s^2\left[\sum_{i=1}^n\nu_i^2+\sum_{i=1}^n\nu_i^2b_i^2\right]\\
&\leq 2ns^2\nu_M\left(\nu_M+P\right),
\end{split}
\]
where we have used our change of random variables in the first equality, and calculated the Gaussian integral resulting from the coordinate-wise expectation value on the second (it converges for $s>-\frac{1}{2\nu_M}$ and therefore for all $s>0$). 
In the first inequality, we applied the logarithmic bound $-\frac12\ln(1+2x)\leq-\frac{x}{1+2x}$ (valid for all $x\geq0$); in the second, we bounded $1+2s\nu_i>1$ which allows us to discard the denominator; and finally, we have bounded the sums $\sum_i\nu_i^2\leq n\nu_M^2$ and $\sum_i\nu_i^2b_i^2\leq\nu_M\|\vec{d}\|^2\leq\nu_MnP$ in the last inequality. With this, the left tail Chernoff bound gives us
\[
\begin{split}
\Pr\{R\leq\mu- \Delta\}&\leq\inf_{s>0}\exp\left[-s\Delta+2ns^2\nu_M\left(\nu_M+P\right)\right]\\
&=\exp\left[\frac{-\Delta^2}{8n\nu_M(\nu_M+P)}\right],
\end{split}
\]
where in the second line we have calculated the minimum over $s>0$. Now, choosing $\Delta:=4n\epsilon(n)^2P-4n\nu_M \sqrt{E_1(n)}>0$, and noticing that $\|\vec{d}\|^2\geq4r^2=4\epsilon(n)^2nP$ we finally obtain:
\begin{equation}\label{eq:E2}
\begin{split}
P_{e,2}(i,j)&=\Pr\{R\leq\mu- \|\vec{d}\|^2+4\nu_M n\sqrt{E_1(n)}\}\\
&\leq\Pr\{R\leq\mu- 4\epsilon(n)^2nP+4\nu_M n\sqrt{E_1(n)}\}\\
&\leq\exp\left[-n\frac{2\left(\epsilon(n)^2P-\nu_M\sqrt{E_1(n)}\right)^2}{\nu_M(\nu_M+P)}\right]\\
&:=e^{-nE_2(n)}.
\end{split}
\end{equation}
The condition $\Delta>0$ is the feasibility condition needed to have a non-trivial left tail and, with our choice of $\Delta$, it can be re-written as:
\begin{equation}\label{eq:condition}
P\epsilon(n)^2>\nu_M \sqrt{E_1(n)}.
\end{equation}

We now have all the tools to study the tradeoff between rate and reliability in the particular distance decoding construction above. One way to do this is by first defining different type I error settings (by suitably choosing $E_1(n)$). Then, the condition in Eq.~\eqref{eq:condition} will lower bound $\epsilon(n)$ which we will plug into Eq.~\eqref{eq:code_size}, obtaining a lower bound on the size of the code. Finally, through Eq.~\eqref{eq:E2} we will study the scaling of the second type of error. 


\begin{theorem}\label{thm:ach_linear}
Let $\nu_M$ be the maximum eigenvalue of the covariance matrix $\Sigma$, and $P$ the maximum power constraint of the linear Gaussian channel. Then, for any arbitrarily small $\tau>0$ we can construct a DI code with linear rate lower bounded as:
\[
R>\log\left(\sqrt{\frac{P}{4\nu_M(1+\tau) \sqrt{E_1}}}-\frac12\right),
\]
and exponentially fast vanishing errors $\lambda_1:=e^{-nE_1}$ with $0<\sqrt{E_1}< \min\{1,\frac{P}{9\tau\nu_M}\}$, and $\lambda_2:=e^{-nE_2}$ with $E_2=\frac{\tau^2E_1}{1+P/\nu_M}$.
\end{theorem}

\begin{proof}
We follow the distance decoding construction developed in the present section. Given any $\tau>0$, let us start by fixing a constant $0<E_1\leq1$ (recall that we have fixed this upper bound in our construction just after Eq.~\eqref{eq:delta_value}, see Remark \ref{cor:big_exp} for the case where $E_1\geq1$), i.e. exponentially fast vanishing first type of error. The condition in Eq.~\eqref{eq:condition} tells us that $\epsilon^2>\nu_M\sqrt{E_1}/P$, so let us choose $\epsilon^2:=(1+\tau)\nu_M\sqrt{E_1}/P$. We have determined in the discussion just after Eq.~\eqref{eq:code_size} that $0<\epsilon<\frac13$, so we need the extra condition $\sqrt{E_1}<\frac{P}{9\tau\nu_M}$. Now, we can calculate the size of the code with Eq.~\eqref{eq:code_size}:
\begin{equation*}
N\geq\left(\frac{1}{2\epsilon}-\frac12\right)^n\geq\left(\frac{\sqrt{P}}{2\sqrt{\nu_M(1+\tau)\sqrt{E_1}}}-\frac12\right)^n,
\end{equation*}
and he rate calculation is completed by simply applying the linear rate definition $R:=\frac1n\log N$.
It is only left to study the second error exponent. We have $\epsilon(n)^2P-\nu_M\sqrt{E_1(n)}=\tau\nu_M\sqrt{E_1(n)}$, so using Eq.~\eqref{eq:Pe2} we get:
\[
E_2=\frac{2\nu_M^2\tau^2E_1}{\nu_M(\nu_M+P)}=\frac{2\tau^2E_1}{1+\frac{P}{\nu_M}},
\]
and the proof is completed.
\end{proof}

\begin{remark}\label{cor:big_exp}
We can recycle the construction above for bigger first type of error exponents $E_1>1$. Notice that, in this regime, the bound in Eq.~\eqref{eq:delta_value} is not true. However, we can easily correct it with $\delta_n\leq4n\nu_ME_1(n)$ [changing $\sqrt{E_1(n)}\to E_1(n)$]. The rest of the construction and the rate result in Theorem \ref{thm:ach_linear} can then also be reused just by equally changing all $\sqrt{E_1(n)}$ by $E_1(n)$. The error exponent conditions in Theorem \ref{thm:ach_linear} become $1<E_1<\frac{P}{9\tau\nu_M}$ and $E_2=\frac{\tau^2E_1^2}{1+P/\nu_M}$.
\end{remark}

We have proven that the explicit distance decoding construction described in the beginning of this section can indeed achieve linear codes with exponentially fast vanishing errors. We will now show that by tuning the construction parameters we can also obtain linearithmic codes with sub-exponentially vanishing errors:
\begin{theorem}\label{thm:achievability_linearithmic}
Given $\tau>0$ and $0<\beta<1$ we can construct a DI code of linearithmic size with linearithmic rate
\[
\frac{R(n)}{\log n}:=\frac{\log N}{n\log n}\geq\frac{\beta}{4}+O\left(\frac{1}{\log n}\right),
\]
and sub-exponentially fast vanishing errors $\lambda_1:=e^{-nE_1(n)}$ and $\lambda_2:=e^{-nE_2(n)}$, with $E_1(n)=n^{-\beta}$ and $E_2(n)=\frac{2\tau^2n}{1+P/\nu_M}n^{-\beta}$.
\end{theorem}

\begin{proof}
Let us start by defining a polynomially-fast vanishing first type error exponent $E_1(n):=n^{-\beta}$ for $0<\beta<1$, which implies a sub-exponentially fast vanishing first type of error, and choosing $\epsilon(n)^2=(1+\tau)\nu_M\sqrt{E_1(n)}/P$ such that the condition in Eq.~\eqref{eq:condition} is respected. Notice that it is the same $\epsilon$ we chose in the proof of Theorem \ref{thm:ach_linear} but with the new definition of $E_1(n)=n^{-\beta}$, therefore, the exponent of the second type of error is also equal $E_2=\frac{2\tau^2E_1(n)}{1+P/\nu_M}$. It is only left to calculate the size of the code through Eq.~\eqref{eq:code_size}:
\[
N\geq\left(\frac{n^{\frac{\beta}{4}}\sqrt{P}}{2\sqrt{(1+\tau)\nu_M}}-\frac12\right)^n\geq\left(\frac{n^{\frac{\beta}{4}}\sqrt{P}}{4\sqrt{(1+\tau)\nu_M}}\right)^n.
\]
Then, the linear rate can be written as 
\[
\begin{split}
R:=\frac{\log N}{n}&\geq\log\left(\frac{n^{\frac{\beta}{4}}\sqrt{P}}{4\sqrt{(1+\tau)\nu_M}}\right)\\   
&=\frac{\beta}{4}\log(n)+\frac12\log\left(\frac{P}{16(1+\tau)\nu_M}\right).
\end{split}
\]
In other words, the linearithmic rate is $\frac{\log N}{n\log n}\geq\frac{\beta}{4}+O(\frac{1}{\log n})$. The proof is completed.
\end{proof}

Theorem \ref{thm:achievability_linearithmic} above shows that linearithmic codes can be constructed with errors that vanish at a sub-exponential speed. 
Furthermore, its analysis on the limit of $n\to\infty$ translates to a capacity lower bound $\dot{C}_{\text{DI}}\geq\frac14$ [see Eq.~\eqref{eq:capacity} and take $\beta\to1$]. 

One could wonder if smaller values of the error exponent (vanishing even faster than polynomially) could result in better rates. The answer is that, with the construction proposed here, we can only improve the sub-leading terms and not the first order of the rate. Indeed, repeating the methods of the proofs above, and remembering that the error exponents must be $E_{1,2}(n)=\omega(1/n)$ so that the errors vanish in the limit $\lim_{n\to\infty}\exp[-nE_{1,2}(n)]=0$, we observe from the condition in Eq.~\eqref{eq:condition} that the minimum order of epsilon we can have is $\epsilon(n)=\omega(n^{-1/4})$, which results on a leading order of $\frac14$ in the linearithmic rate. 

To conclude the achievability analysis let us add a note on the maximally asymmetric error regime case. Clearly, the converse results in Section \ref{ssec:asymmetric} and the matching linear scale code construction in Theorem \ref{thm:ach_linear} are enough to prove that DI codes over Gaussian channels in the maximally asymmetric error regimes behave linearly. Indeed, any code for exponentially fast vanishing errors is also a code for the (more permissive) Stein and Sanov regimes. 
Further intuition for the achievability in these maximal asymmetric cases can be gained by playing with different error settings in our construction, studying the condition in Eq.~\eqref{eq:condition}, the lower bound on the code size in Eq.~\eqref{eq:code_size} and $E_2(n)$ [as defined in Eq.~\eqref{eq:Pe2}]. The reader will readily see that, as already proved with our converse results, we can do no better than linear rates with our distance decoding construction unless \emph{both} errors vanish sub-exponentially.

\medskip
\section{Conclusions}
We have initiated the study of the rate-reliability tradeoff for deterministic identification over channels with continuous output. Specifically, we have tackled the general linear Gaussian channel, motivated by its vast range of potential applications. 

The results of this work show a strong resemblance with those for the general family of channels with discrete output, studied in \cite{CDBW:Reliability-TCOM}. We have shown that when the errors vanish exponentially fast, the typical linearithmic behaviour of DI over Gaussian channels is lost (Subsection \ref{ssec:converse}), and instead at most a linear rate of DI is possible, bounded in terms of the logarithm of the inverse exponent. As a matter of fact, this superlinear scaling is lost even if only one of the two errors vanishes exponentially fast, as seen in Subsection \ref{ssec:asymmetric}, mimicking the previous results from \cite{CDBW:Reliability-TCOM}. We also show how the linearithmic behaviour is recovered for slowly (sub-exponential) vanishing errors, and present an explicit code construction that matches the converse results in different error settings.

It is curious to study the nature of the rates in the different error scalings. In linear codes, we observe that the rates depend on channel parameters like the power constraint and the maximum eigenvalue of the covariance matrix at leading order (see Theorems~\ref{thm:RR_upperbound}, \ref{thm:stein_upperbound}, and \ref{thm:ach_linear}). However, in linearithmic codes, this dependence on channel parameters is lost at leading order and relegated to lower orders (see Theorems~\ref{thm:RR_upperbound}, and \ref{thm:stein_upperbound} with $E(n)=\Omega(1/n)$, and Theorem~\ref{thm:achievability_linearithmic}). Furthermore, in non-deterministic identification, where the length of the words can be exponential, the dependence on channel parameters appears again at leading order \cite{DI-fading,LDB:secure_ID_Gaussian}. Indeed, linear capacities for DI (where we can have exponentially fast vanishing errors) and exponential capacities for randomized identification over Gaussian channels (which coincide with the Shannon transmission linear capacities) are functions of the power constraint and certain relevant parameters of the covariance matrix; while linearithmic capacities for DI are just a real number, independent of $P$ and $\Sigma$. 
This is further evidence of the strange nature of DI in the linearithmic scale where, for example in general channels with discrete output, the linearithmic capacity depends on a complexity (fractal) measure of the output probability set, the Minkowski dimension, a scale invariant object which behaves very differently from the metric objects that usually define capacities in communication settings (like the divergence radius in Shannon's communication capacity, or the distortion metric in rate-distortion theory).


Given the present results and the already discussed similarity to general channels with discrete output, one might be led to expect finding similar results for general channels with arbitrary continuous output, beyond the Gaussian case studied here. How to do this, however, is not obvious. Indeed, in the present paper, we have managed to move from necessary requirements between output distributions to Euclidean distances between the corresponding inputs, allowing us to use volumetric arguments to complete the proofs. We can, of course, only use this tool thanks to the particular characteristics and symmetries of the channel, but it is not available in general. 
It is worth noting that for general channels with discrete output, we could move beyond the volumetric argument by observing that a natural metric arises for general product distributions over a finite alphabet (see \cite[Sect.~IV]{CDBW:DI_classical}), and we could then use packing and covering arguments on this metric to develop the proofs. Unfortunately, this general metric cannot be translated to product distributions on continuous spaces. Hence, as covering and packing numbers depend directly on the metric we choose, we can not reuse the ideas from the discrete case in the general continuous scenario. In other words, the lack of a general metric would force us to choose different metrics for different families of channels with continuous output, resulting in different results if covering and packing arguments were used. 

In any case, and as it can already be intuited from the discussion above, if one finds a good metric for a particular family of continuous channels, then the general methods developed in \cite{CDBW:DI_classical,CDBW:Reliability-TCOM} would be available to yield at least converse results for the capacity and rate-reliability upper bounds, showing the power of our abstract methods.

\bigskip
\noindent
{
\textbf{Acknowledgments.} 
PC and AW are supported by the Institute for Advanced Study of the Technical University Munich. 
HB, CD and PC acknowledge financial support from the Federal Ministry of Education and Research (BMBF) within the program ``Souver\"an.\ Digital.\ Vernetzt.'' as part of the joint project 6G-life (project IDs 16KISK002 and 16KISK263), as well as from the BMBF quantum programs QuaPhySI (grants 16KIS1598K and 16KIS2234), QUIET (grants 16KISQ093 and 16KISQ0170), and QD-CamNetz (grants 16KISQ077 and 16KISQ169). Further support was provided by the German Federal Ministry of Research, Technology and Space (BMFTR) under grants 16KIS2611 and 16KIS2602 (QSTARS), 16KISR027K and 16KISR038 (Q-TREX), and 16KIS2414 and 16KIS2415 (Research Hub 6G-life$^2$). The research is additionally funded by the German Research Foundation (DFG) under the project ``Post Shannon Theory and Implementation'' (grants BO 1734/38-1 and DE 1915/2-1), under grant BO 1734/20-1, and within Germany’s Excellence Strategy --- EXC 2050/2 (Project ID 390696704), Cluster of Excellence ``Centre for Tactile Internet with Human-in-the-Loop'' (CeTI), Technische Universit\"at Dresden. HB and CD further receive funding through the national initiative under grants 16KIS1003K and 16KIS1005, respectively. Support from the Bavarian Ministry of Economic Affairs, Regional Development and Energy within the project ``6G Future Lab Bavaria'' is also gratefully acknowledged.
AW is supported by the European Commission QuantERA project ExTRaQT (Spanish MICIN grant no.~PCI2022-132965); by the Spanish MICIN (project PID2022-141283NB-I00) with the support of FEDER funds; by the Spanish MICIN with funding from European Union NextGenerationEU (PRTR-C17.I1) and the Generalitat de Catalunya; by the Spanish MTDFP through the QUANTUM ENIA project: Quantum Spain, funded by the European Union NextGenerationEU within the framework of the "Digital Spain 2026 Agenda”; and by the Alexander von Humboldt Foundation. 
}

\bigskip

\AtNextBibliography{\small}
\printbibliography

@INPROCEEDINGS{LDB:secure_ID_Gaussian,
  author={Labidi, Wafa and Deppe, Christian and Boche, Holger},
  booktitle={ICASSP 2020 - 2020 IEEE International Conference on Acoustics, Speech and Signal Processing (ICASSP)}, 
  title={Secure Identification for Gaussian Channels}, 
  year={2020},
  volume={},
  number={},
  pages={2872-2876},
  doi={10.1109/ICASSP40776.2020.9054336}}

@inbook{Gallager:Gauss_chapter,
edition={1}, 
title = {Memoryless Channels with Discrete Time},
booktitle = {Information Theory and Reliable Communication},
chapter = {7},
publisher={Springer Verlag},
address = {Vienna},
series = {CIMS -- International Centre for Mechanical Sciences, Courses and Lectures},
volume = {30},
author={Robert G. Gallager}, 
year={1972},
doi = {10.1007/978-3-7091-2945-6}
}

@ARTICLE{AD:ID_ViaChannels,
  author={Ahlswede, Rudolf and Dueck, Gunter},
  journal={IEEE Transactions on Information Theory}, 
  title={Identification via channels}, 
  year={1989},
  volume={35},
  number={1},
  pages={15-29},
  doi={10.1109/18.42172}
}

@ARTICLE{CDBW:DI_classical,
  author={Colomer, Pau and Deppe, Christian and Boche, Holger and Winter, Andreas},
  journal={IEEE Transactions on Information Theory}, 
  title={Deterministic identification over channels with finite output: a dimensional perspective on superlinear rates}, 
  year={2025},
  volume={71},
  number={5},
  pages={3373-3396},
  doi={10.1109/TIT.2025.3531301}
}

@book{CS:Packings_lattices,
author = {Conway, John H. and Sloane, Neil J. A.},
year = {1988},
title = {Sphere Packings, Lattices and Groups},
publisher = {Springer Verlag},
series={Grundlehren der Mathematischen Wissenschaften},
volume = {290},
doi = {10.1007/978-1-4757-2016-7}
}

@ARTICLE{AC:DI,
  author={Ahlswede, Rudolf and Cai, Ning},
  journal={IEEE Transactions on Information Theory}, 
  title={Identification without randomization}, 
  year={1999},
  month={11},
  volume={45},
  number={7},
  pages={2636-2642},
  doi={10.1109/18.796419}
}

@ARTICLE{SPBD:DI_power,
  author={Salariseddigh, Mohammad J. and Pereg, Uzi and Boche, Holger and Deppe, Christian},
  journal={IEEE Transactions on Information Theory}, 
  title={{Deterministic Identification Over Channels With Power Constraints}}, 
  year={2022},
  volume={68},
  number={1},
  month={1},
  pages={1-24},
  doi={10.1109/TIT.2021.3122811}
}

@article{CMW16:Dh-to-D_alpha,
   title={{Strong Converse Exponents for a Quantum Channel Discrimination Problem and Quantum-Feedback-Assisted Communication}},
   volume={344},
   DOI={10.1007/s00220-016-2645-4},
   number={3},
   journal={Communications in Mathematical Physics},
   publisher={Springer Science and Business Media LLC},
   author={Cooney, Tom and Mosonyi, Mil\'an and Wilde, Mark M.},
   year={2016},
   month=may, 
   pages={797–829}
}

@ARTICLE{CDBW:Reliability-TCOM,
  author    = {Colomer, Pau and Deppe, Christian and Boche, Holger and Winter, Andreas},
  title     = {{Rate-Reliability Tradeoff for Deterministic Identification}},
  journal   = {IEEE Transactions on Communications},
  year      = {2025},
  note      = {Early access; also available as arXiv:2502.02389},
  doi       = {10.1109/TCOMM.2025.3594790}
}

@inbook{CT_book,
publisher = {John Wiley \&{} Sons, Ltd},
title = {Gaussian Channel},
booktitle = {Elements of Information Theory},
chapter = {9},
pages = {261-299},
author={Thomas M. Cover and Joy A. Thomas},
doi = {10.1002/047174882X.ch9},
year = {2005}
}

@inproceedings{CDBW:Reliability_ICC,
  title = {Rate-reliability tradeoff for deterministic identification},
  author = {Pau Colomer and Christian Deppe and Holger Boche and Andreas Winter},
  booktitle = {Proc. 2025 IEEE International Conference on Communications, 8–12 June 2025, Montreal, Canada},
  publisher = {IEEE},
  pages = {},
  year = {},
  doi = {}
}

@inproceedings{CDBW:Gaussian_ICC,
  title = {Rate-Reliability Tradeoff for Deterministic Identification over Gaussian Channels},
  author = {Pau Colomer and Christian Deppe and Holger Boche and Andreas Winter},
  booktitle = {Proc. 2026 IEEE International Conference on Communications, 24–28 May 2026, Glasgow, United Kingdom},
  publisher = {IEEE},
  pages = {},
  year = {},
  doi = {}
}

@incollection{6G_Book,
	author = {H. Boche and J.A. Cabrera and C. Deppe and P. Kutsevol and S. Wang and F.H.P. Fitzek and S. Hirche and W. Kellerer and W. Labidi and J. Rosenberger and R.F. Schaefer and C. von Lengerke and M. Wiese and S. Rezwan and P.K.H. Sheshagiri and R. Ezzine},
    editor = {F.H.P. Fitzek and H. Boche and W. Kellerer and P. Seeling},
	title = {Novel Information Theoretical Approaches},
	booktitle = {6G-life: Unveiling the Future of Technological Sovereignty, Sustainability and Trustworthiness},
	year = {2026},
    month = {2},
    publisher = {Academic Press},
    isbn = {978-0443274107},
    language = {en},		
}

@incollection{6G_Book_2,
	author = {F.H.P. Fitzek and P. Schwenteck and H. Boche and W. Kellerer and G.T. Nguyena and P. Seeling},
    editor = {F.H.P. Fitzek and H. Boche and W. Kellerer and P. Seeling},
	title = {6G Perspective of Mobile Network Operators, Manufacturers, and Verticals},
	booktitle = {6G-life: Unveiling the Future of Technological Sovereignty, Sustainability and Trustworthiness},
	year = {2026},
    month = {2},
    publisher = {Academic Press},
    isbn = {978-0443274107},
    language = {en},	
}

@article{HanVerdu:ID,
  author={Han, Te Sun and Verd\'u, Sergio},
  title={New results in the theory of identification via channels}, 
  year={1992},
  month={1},
  journal = {IEEE Transactions on Information Theory},
  volume={38},
  number={1},
  pages={14-25},
  doi={10.1109/18.108245}
}

@INPROCEEDINGS{DI-fading,
  author={Salariseddigh, Mohammad J. and Pereg, Uzi and Boche, Holger and Deppe, Christian},
  booktitle={Proc. 2020 IEEE Information Theory Workshop (ITW)}, 
  title={{Deterministic Identification Over Fading Channels}}, 
  year={2021},
  pages={1-5},
  doi={10.1109/ITW46852.2021.9457587}
}

@INPROCEEDINGS{DI-steins,
  author={Colomer, Pau and Deppe, Christian and Boche, Holger and Winter, Andreas},
  booktitle={Proc. 2025 IEEE European Wireless}, 
  title={Deterministic Identification in Maximally Asymmetric Error Regimes}, 
  year={2025},
  pages={},
  doi={}
}

@ARTICLE{DI-poisson_mc,
  author={Salariseddigh, Mohammad J. and Jamali, Vahid and Pereg, Uzi and Boche, Holger and Deppe, Christian and Schober, Robert},
  journal={IEEE Transactions on Molecular, Biological and Multi-Scale Communications}, 
  title={{Deterministic Identification for Molecular Communications Over the Poisson Channel}}, 
  year={2023},
  month={4},
  volume={9},
  number={4},
  pages={408-424},
  doi={10.1109/TMBMC.2023.3324487}}

@unpublished{galaxy-codes,
    author = {Holger Boche and Christian Deppe and Safieh Mahmoodi and Golamreza Omidi},
    title = {{Galaxy Codes: Advancing Achievability for Deterministic Identification via Gaussian Channels}},
    note = {arXiv[cs.IT]:2501.12548},
    year={2025},
    month={1},
    doi={10.48550/arXiv.2501.12548}
}

@book{LieseVajda,
    author = {Friedrich Liese and Igor Vajda},
    title = {Convex Statistical Distances},
    publisher = {B. G. Teubner Verlagsgesellschaft},
    address = {Leipzig},
    year = {1987},
    series = {Texte zur Mathematik},
    volume = {95},
    isbn = {3-322-00428-7}
}

\end{document}